\documentclass[a4paper,10pt]{paper}
\usepackage{mathrsfs}
\usepackage{amsmath,amssymb,amsthm}
\usepackage{cite}
\usepackage{color}
\usepackage{hyperref}
\usepackage{tikz}

\newtheorem{proposition}{\bf Proposition}[section]
\numberwithin{equation}{section}

\newcommand{\bse}{\begin{subequations}}
\newcommand{\ese}{\end{subequations}}

\DeclareMathOperator{\diag}{Diag}
\def\undertilde#1{\mathord{\vtop{\ialign{##\crcr
$\hfil\displaystyle{#1}\hfil$\crcr\noalign{\kern1.5pt\nointerlineskip}
$\hfil\widetilde{}\hfil$\crcr\noalign{\kern-6.5pt}}}}}
\def\underhat#1{\mathord{\vtop{\ialign{##\crcr
$\hfil\displaystyle{#1}\hfil$\crcr\noalign{\kern1.5pt\nointerlineskip}
$\hfil\widehat{}\hfil$\crcr\noalign{\kern-6.5pt}}}}}
\def\underbar#1{\mathord{\vtop{\ialign{##\crcr
$\hfil\displaystyle{#1}\hfil$\crcr\noalign{\kern1.5pt\nointerlineskip}
$\hfil\bar{}\hfil$\crcr\noalign{\kern-6.5pt}}}}}

\def\wt#1{\widetilde{#1}}
\def\wh#1{\widehat{#1}}
\def\wb#1{\bar{#1}}
\def\th#1{\wh{\wt{#1}}}
\def\hb#1{\wb{\wh{#1}}}
\def\bt#1{\wt{\wb{#1}}}
\def\thb#1{\wb{\wh{\wt{#1}}}}
\def\ut#1{\undertilde{#1}}
\def\uh#1{\underhat{#1}}
\def\ub#1{\underbar{#1}}
\def\uth#1{\ut{\uh{#1}}}

\newcommand{\rT}{\mathrm{T}}
\newcommand{\rd}{\mathrm{d}}

\newcommand{\cN}{\mathcal{N}}
\newcommand{\cF}{\mathcal{F}}

\newcommand{\bOa}{\mathbf{\Omega}}

\newcommand{\bC}{\mathbf{C}}

\newcommand{\bU}{\mathbf{U}}

\newcommand{\bu}{\mathbf{u}}

\newcommand{\bc}{\mathbf{c}}
\newcommand{\tbc}{{}^{t\!}\mathbf{c}}
\newcommand{\bA}{\mathbf{A}}

\newcommand{\bM}{\mathbf{M}}

\newcommand{\ld}{\lambda}

\newcommand{\oa}{\omega}
\newcommand{\Oa}{\Omega}

\title{On reductions of the discrete Kadomtsev--Petviashvili-type equations}
\author{Wei FU and Frank W NIJHOFF \\
 School of Mathematics, University of Leeds, Leeds LS2 9JT, UK \\
}

\begin{document}

\maketitle

\begin{abstract}
The reduction by restricting the spectral parameters $k$ and $k'$ on a generic algebraic curve of degree $\mathcal{N}$ is performed for the discrete AKP, 
BKP and CKP equations, respectively. A variety of two-dimensional discrete integrable systems possessing a more general solution structure arise from the 
reduction, and in each case a unified formula for generic positive integer $\mathcal{N}\geq 2$ is given to express the corresponding reduced integrable 
lattice equations. The obtained extended two-dimensional lattice models give rise to many important integrable partial difference equations as special 
degenerations. Some new integrable lattice models such as the discrete Sawada--Kotera, Kaup--Kupershmidt and Hirota--Satsuma equations in extended form 
are given as examples within the framework. 
\end{abstract}

\section{Introduction}\label{S:Intro}
Discrete equations are often considered to play the role of master models in the theory of integrable systems due to their rich algebraic structure, 
namely a single discrete integrable equation encodes the information of the whole hierarchy of the corresponding continuous equation. More importantly, 
the integrability of discrete equations also arises in modern mathematics and physics (see e.g. \cite{HJN16}). A key feature of discrete integrable systems 
is the property of the so-called \emph{multi-dimensional consistency} (MDC) -- a lattice equation can be consistently embedded into a higher-dimensional 
lattice space \cite{DS97,NW01} (see also \cite{ABS03,ABS12}). An even deeper understanding is that the consistent discrete equations are also compatible 
with their corresponding hierarchy of continuous equations; more precisely, discrete and continuous independent variables are on the same footing and 
consequently both discrete and continuous equations can be simultaneously embedded into an infinite-dimensional space spanned by both discrete and 
continuous coordinates. Following this idea, by integrability of an equation we mean that there exists a common explicit solution admitting infinite 
degrees of freedom for both discrete and continuous hierarchies. 

Studying three-dimensional (3D) integrable lattice models might be a very proper and effective path to understand the theory of integrable systems -- the 
higher dimension brings richer structure and the integrability structure of two-dimensional (2D) discrete systems is hidden in that of 3D lattice equations. 
Among 3D discrete integrable systems, there are three typical lattice equations of the Kadomtsev--Petviashvili-type (KP), namely the discrete AKP \cite{Hir81}, 
BKP \cite{Miw82} and CKP \cite{Kas96,Sch03} equations, where the letters `A', `B' and `C' follow from the Kyoto school's transformation group theory of 
soliton equations, see \cite{JM83} and references therein. The discrete KP equations are shown to possess the MDC property \cite{ABS12,HJN16,TW09,Atk11}, 
and this can also be understood under a unified framework for both discrete and continuous KP hierarchies \cite{FN16,FN17}. The continuous KP hierarchies 
play the role of master models in the continuous theory since they reduce to all the soliton hierarchies associated with scalar differential spectral problems, 
see e.g. \cite{JM83} and \cite{FN17}. Therefore it is an interesting question to understand how reductions work for the discrete KP equations and consequently 
to establish the corresponding discrete theory. 

Reductions of the discrete KP equations apart from the discrete CKP equation have been considered. One approach is to perform reductions on the 
$\tau$-functions in their bilinear form, namely simple constraints can be imposed on the $\tau$-functions describing the bilinear 3D lattice 
equations and as a result many 2D discrete integrable equations such as the the discrete Korteweg--de Vries (KdV), Boussinesq (BSQ), Sawada--Kotera (SK) 
equations are obtained (see e.g. \cite{DJM83,TH96,HZ13}). The idea behind such a procedure is actually a reduction on the two spectral parameters $k$ and $k'$, 
namely they are related through an algebraic curve. However, in those works, some fixed lattice parameters are involved in the algebraic curve, i.e. 
the preselected discrete variables and their associated lattice parameters are treated specifically compared with the others. This implies that the key 
point that all the independent variables must be on the same footing is broken, and consequently solution structures of the reduced 2D lattice models 
become more restrictive. Instead, a trilinear discrete BSQ equation as an example was derived from the \emph{direct linearisation} (DL) approach recently 
\cite{ZZN12}, where the associated algebraic curve is a generic one (independent of lattice parameters). The trilinear discrete BSQ equation possesses the 
most general solution structure since no restriction on the lattice parameters is required. This observation motivates us to reconsider reductions of the 
discrete KP equations generically. 

In a recent paper \cite{FN16}, the DL framework was established for the discrete AKP, BKP and CKP equations and their integrability. This unified framework 
enables us to study reductions of the discrete KP equations from this view point in the current paper. The starting point is a generic algebraic curve taking 
the form of 
\begin{align}\label{AlgCurve}
C(t,k)\doteq P(t)-P(k)=0, \quad \hbox{where} \quad P(t)=\sum_{j=1}^\cN\alpha_jt^j \quad \hbox{for} \quad \cN=2,3,\cdots. 
\end{align}
Without losing generality, we take $\alpha_\cN=1$, namely the polynomial $P(k)$ is monic. Since $P(t)$ is a polynomial of degree $\cN$, $P(t)-P(k)$ can be 
theoretically factorised as $\prod_{j=1}^\cN(t-\oa_j(k))$ where $\oa_j(k)$ depend on the parameters $\alpha_j$. This implies that Equation \eqref{AlgCurve} 
as an algebraic equation of $t$ has $\cN$ roots $t=\oa_j(k)$. It is worth pointing out that the curve degenerates when $\alpha_j=0$ for $j=1,2,\cdots,\cN-1$ 
and in this case we have the roots $t=\oa^j k$ for $j=1,2,\cdots$, where $\oa$ is one of the primitive $\cN$th roots of unity defined by 
$\oa=\exp(2\pi\mathfrak{i}/\cN)$ with $\mathfrak{i}$ the imaginary unit. In order to avoid the ambiguity in the future, we require that 
$\oa_j(k)|_{\alpha_1=\cdots=\alpha_{\cN-1}=0}=\oa^j k$ for $j=1,2,\cdots,\cN$. Hence the generic curve \eqref{AlgCurve} has $\cN-1$ primitive roots 
$t=\oa_j(k)$ for $j=1,2,\cdots,\cN-1$ (because more parameters $\alpha_j$ are introduced in the curve when the degree $\cN$ increases) and the non-primitive 
root $\oa_\cN(k)=k$ as it solves \eqref{AlgCurve} for any $\cN$. 

The reduction is performed by restricting the spectral parameters $k$ and $k'$ on the algebraic curve $C(-k',k)=0$, namely by taking any of the primitive 
roots $k'=-\oa_j(k)$ with $j=1,2,\cdots,\cN-1$ (or their combination as we will see later). As a result, 2D integrable lattice models in their extended form 
in terms of the extra parameters $\alpha_j$ arise from such a reduction where the extra parameters $\alpha_j$, $j=1,2,\cdots,\cN-1$ are introduced into the 
discrete systems. The degeneration $\alpha_j=0$ gives rise to $k'=-\oa^jk$, which is corresponding to the unextended discrete systems. One can observe that 
no lattice parameters are involved in such a reduction and therefore the 2D lattice models possess the most general solution structure. For each KP-type 
equation, a unified expression (coupled system form) for the reduced 2D integrable lattice equations in extended form is given for generic positive integer 
$\cN\geq 2$. Examples include the discrete KdV and BSQ equations, a new discretisation of the SK equation, and some 2D discrete integrable systems which 
have not been investigated in the literature before, such as the discrete Kaup--Kupershmidt (KK) and Hirota--Satsuma (HS) equations. 

The paper is organised as follows. In Section \ref{S:DIS}, we briefly review the general structure of discrete and continuous integrable hierarchies, and 
explain the MDC property of integrable lattice equations, from the view point of the DL. In Section \ref{S:Comparison}, a comparison between the old and 
new reductions is made, illustrated by the discrete BSQ and KdV equations. The main result, namely a general reduction scheme, is proposed in Section 
\ref{S:Reduc}, and examples including the discrete BSQ, SK, KK and HS equations in extended form together with their general solution structure are also 
given in the section. 

\section{General structure of discrete and continuous integrable hierarchies}\label{S:DIS}

Let us first introduce some short-hand notations in the discrete theory for convenience. Suppose $f\doteq f(n_1,n_2,n_3,\cdots)$ is a function of 
discrete variables $\{n_1,n_2,n_3,\cdots\}$ associated with their corresponding parameters $\{p_1,p_2,p_3,\cdots\}$, the discrete forward and backward 
shift operators are defined by 
\begin{align*}
 \rT_{p_\gamma}f(n_1,\cdots,n_\gamma,\cdots)\doteq f(n_1,\cdots,n_\gamma+1,\cdots), \quad 
 \rT_{p_\gamma}^{-1}f(n_1,\cdots,n_\gamma,\cdots)\doteq f(n_1,\cdots,n_\gamma-1,\cdots)
\end{align*}
for $\gamma=1,2,\cdots$. Conventionally, we identify the discrete variables $n_1\equiv n$, $n_2\equiv m$, $n_3\equiv h$ and the associated lattice parameters 
$p_1\equiv p$, $p_2\equiv q$, $p_3\equiv r$. As a result, some short-hand notations are introduced as follows: 
\begin{align*}
 \wt f\doteq \rT_p f, \quad 
 \wh f\doteq \rT_q f, \quad 
 \wb f\doteq \rT_r f, \quad 
 \ut f\doteq \rT_p^{-1} f, \quad 
 \uh f\doteq \rT_q^{-1} f, \quad 
 \ub f\doteq \rT_r^{-1} f.
\end{align*}
Combinations of the tilde-, hat- and bar-signs (which will be used everywhere in the rest of this paper) denote the compositions of their corresponding 
discrete shift operations. 
\subsection{Direct linearisation}\label{S:DL}

From the view point of the DL, 3D discrete and continuous integrable hierarchies are governed by a linear integral equation taking the form of 
\begin{align}\label{IntEq}
 \bu_k+\iint_D\rd\zeta(l,l')\rho_k\Oa_{k,l'}\sigma_{l'}\bu_l=\rho_k\bc_k,
\end{align}
where $\bu_k$ and $\bc_k$ are infinite column vectors having components $u_k^{(i)}=u_k^{(i)}(n_1,n_2,n_3,\cdots;x_1,x_2,x_3,\cdots;k)$ and $c_k^{(i)}=k^i$, 
respectively. The plane wave factors $\rho_k$ and $\sigma_{l'}$ as functions of the discrete independent variables $\{n_\gamma|\gamma\in\mathbb{Z}^+\}$ 
and their associated lattice parameters $\{p_\gamma|\gamma\in\mathbb{Z}^+\}$, as well as the continuous independent variables 
$\{x_\gamma|\gamma\in\mathbb{Z}^+\}$, depend on the spectral variables $k$ and $l'$, respectively, and they describe the linear behaviour of the 
corresponding nonlinear integrable equations, and the Cauchy kernel $\Oa_{k,l'}$ as an expression of only the spectral variables $k$ and $l'$ reflects the 
nonlinear behaviour. 

The key point is the separation of $l$ and $l'$ in the measure on the integration domain $D$ which is a reflection of an underlying nonlocal Riemann--Hilbert 
(or $\bar\partial$-) problem leading to 3D integrable hierarchies -- the term `nonlocal Riemann--Hilbert problem' was introduced by Zakharov and Manakov, see 
e.g. \cite{ZM85}; in other words, a double integral is necessary in the linear integral equation for 3D discrete and continuous integrable equations. In 
addition, both discrete and continuous independent variables here should be considered being on the same footing and therefore for the most general structure 
no restriction can be made on the lattice parameters. For the same reason, the Cauchy kernel $\Oa_{k,l'}$ must be independent of the lattice parameters as well. 

In the DL method, by selecting a suitable Cauchy kernel $\Oa_{k,l'}$, and its compatible plane wave factors $\rho_k$ and $\sigma_{l'}$, as well as a proper 
measure $\rd\zeta(l,l')$ in the linear integral equation \eqref{IntEq}, a class of discrete and continuous integrable hierarchies can be constructed from 
the underlying structure. The resulting discrete and continuous equations are integrable in the sense of possessing Lax pairs as well as explicit solutions. 
More precisely, for given $\Oa_{k,k'}$, $\rho_k$ and $\sigma_{k'}$, we can define an operator $\bOa$ and infinite matrices $\bC$ and $\bU$ by the following: 
\begin{align}\label{Sol}
 \Oa_{k,k'}=\tbc_{k'}\cdot\bOa\cdot\bc_k, \quad \bC=\iint_D\rd\zeta(k,k')\rho_k\bc_k\tbc_{k'}\sigma_{k'}, \quad 
 \bU=\iint_D\rd\zeta(k,k')\bu_k\tbc_{k'}\sigma_{k'},
\end{align}
where $\tbc_{k'}$ is an infinite row vector as the transpose of $\bc_{k'}$. Then the linear problems (Lax pairs), the nonlinear forms and the 
bilinear/multilinear forms of the corresponding discrete and continuous hierarchies can be recovered from the quantities $\bu_k$, $\bU$ and 
$\det(1+\bOa\cdot\bC)\doteq\tau$ (or sometimes $\tau^2$), respectively, within the DL framework. For more details, we refer the reader to \cite{FN16,FN17}. 

The double integral in $\bU$ and $\bC$ provides more general solution structure of the 3D integrable lattice and continuous equations. Suitable choices of 
the measure can indeed bring us particular classes of explicit solutions to the nonlinear/multilinear equations, e.g. soliton solutions. For the sake of 
convenience, in the paper we only discuss the integrability structure of a discrete equation from the perspective of the $\tau$-function as it provides 
sufficient information to understand its explicit solution and the MDC property.

\subsection{Discrete and continuous Kadomtsev--Petviashvili-type hierarchies}\label{S:KP}

The three typical scalar 3D integrable models, namely the discrete and continuous AKP, BKP and CKP, can be recovered by choosing the following Cauchy kernels, 
plane wave factors as well as measures. Since this paper is contributed to reductions of the discrete KP-type equations, we do not list the continuous equations 
here, but the continuous structures will still be given for comparison. 

\paragraph{AKP hierarchy.} For the discrete and continuous AKP hierarchies, we choose the Cauchy kernel and the plane wave factors as follows: 
\begin{align}\label{AKP:Info}
 \Oa_{k,k'}=\frac{1}{k+k'}, \,
 \rho_k=\prod_{\gamma=1}^{\infty}(p_\gamma+k)^{n_\gamma}\exp\left(\sum_{\gamma=1}^{\infty}k^\gamma x_\gamma\right), \,
 \sigma_{k'}=\prod_{\gamma=1}^{\infty}(p_\gamma-k')^{-n_\gamma}\exp\left(-\sum_{\gamma=1}^{\infty}(-k')^\gamma x_\gamma\right),
\end{align}
and in this case, there is no particular requirement for the measure $\rd\zeta(k,k')$. The $\tau$-function for soliton solution given by 
\begin{align}\label{AKP:Sol}
 \tau=\det(1+\bA\bM), \quad (\bM)_{j,i}=M_{j,i}=\frac{\rho_{k_i}\sigma_{k'_j}}{k_i+k'_j}, \quad i=1,2,\cdots,N,\quad j=1,2,\cdots,N',
\end{align}
where $\bA$ with entries $A_{i,j}$ is an arbitrary $N\times N'$ non-degenerate matrix, solves simultaneously the discrete, semi-discrete and continuous 
AKP hierarchies, see\cite{FN16,FN17}. The $\tau$-function obeys an infinite copies of the same 3D discrete equation up to lattice directions and parameters 
-- this is the MDC property (which will be explained later) and the whole family of discrete equations should be understood as the discrete AKP hierarchy. 
Without loss of generality, we write down the discrete equation involving the lattice directions $n_1$, $n_2$ and $n_3$ and it takes the form of  
\begin{align}\label{dAKP}
 (p-q)\wb\tau\th\tau+(q-r)\wt\tau\hb\tau+(r-p)\wh\tau\bt\tau=0.
\end{align}
This is the bilinear discrete AKP equation, which was originally proposed by Hirota \cite{Hir81} and was entitled the discrete analogue of the generalised 
Toda equation (DAGTE). And the equation is now also known as the Hirota--Miwa equation due to Miwa's observation that all the members in the continuous KP 
hierarchy are actually embedded in this single equation \cite{Miw82}. 
\paragraph{BKP hierarchy.}
For the discrete BKP hierarchy, we choose the Cauchy kernel and the plane wave factors and impose the antisymmetry property on the measure as follows:
\begin{align}\label{BKP:Info}
 \Oa_{k,k'}=\frac{1}{2}\frac{k-k'}{k+k'}, \,
 \rho_k=\prod_{\gamma=1}^{\infty}\left(\frac{p_\gamma+k}{p_\gamma-k}\right)^{n_\gamma}\exp\left(\sum_{\gamma=0}^{\infty}k^{2\gamma+1}x_{2\gamma+1}\right), \,
 \sigma_{k'}=\rho_{k'}, \, \rd\zeta(k,k')=-\rd\zeta(k',k).
\end{align}
Then the DL framework reveals that the $\tau$-function for soliton solution given by 
\begin{align}\label{BKP:Sol}
 \tau^2=\det(1+\bA\bM), \quad (\bM)_{j,i}=M_{j,i}=\rho_{k_i}\frac{1}{2}\frac{k_i-k'_j}{k_i+k'_j}\sigma_{k'_j}, \quad A_{i,j}=-A_{j,i}, \quad i,j=1,2,\cdots,2N,
\end{align}
solves the discrete, semi-discrete as well as continuous BKP hierarchies (see \cite{FN16} for the discrete one and \cite{FN17} for the continuous one). 
The discrete equation describing the dynamical evolutions with respect to the lattice directions $n_1$, $n_2$ and $n_3$ takes the form of 
\begin{align}\label{dBKP}
 &(p-q)(q-r)(r-p)\tau\thb\tau+(p+q)(p+r)(q-r)\wt\tau\hb\tau \nonumber \\
 &\quad+(r+p)(r+q)(p-q)\wb\tau\th\tau+(q+r)(q+p)(r-p)\wh\tau\bt\tau=0,
\end{align}
which is usually referred to as the Miwa equation \cite{Miw82}. In addition, the $\tau$-function can be expressed by a Pfaffian due to the antisymmetry of 
the matrices $\bA$ and $\bM$. 

\paragraph{CKP hierarchy.}
The following Cauchy kernel, plane wave factors and measure are chosen for the discrete and continuous CKP hierarchies: 
\begin{align}\label{CKP:Info}
 \Oa_{k,k'}=\frac{1}{k+k'}, \,
 \rho_k=\prod_{\gamma=1}^{\infty}\left(\frac{p_\gamma+k}{p_\gamma-k}\right)^{n_\gamma}\exp\left(\sum_{\gamma=0}^{\infty}k^{2\gamma+1}x_{2\gamma+1}\right), \,
 \sigma_{k'}=\rho_{k'}, \,  \rd\zeta(k,k')=\rd\zeta(k',k).
\end{align}
The $\tau$-function for soliton solution in this case is given by 
\begin{align}\label{CKP:Sol}
 \tau=\det(1+\bA\bM), \quad (\bM)_{j,i}=M_{j,i}=\frac{\rho_{k_i}\sigma_{k'_j}}{k_i+k'_j}, \quad A_{i,j}=A_{j,i}, \quad i,j=1,2,\cdots,N,
\end{align}
and it was proven in \cite{FN16,FN17} that the $\tau$-function solves the discrete, semi-discrete and continuous CKP hierarchies. An example 
involving the lattice variables $n_1$, $n_2$ and $n_3$ takes the form of Cayley's $2\times2\times2$ hyperdeterminant as follows: 
\begin{align}\label{dCKP}
 &\Big[(p-q)^2(q-r)^2(r-p)^2\tau\thb\tau+(p+q)^2(p+r)^2(q-r)^2\wt\tau\hb\tau \nonumber \\
 &\quad-(q+r)^2(q+p)^2(r-p)^2\wh\tau\bt\tau-(r+p)^2(r+q)^2(p-q)^2\wb\tau\th\tau\Big]^2 \nonumber \\
 &\qquad-4(p^2-q^2)^2(p^2-r^2)^2\Big[(q+r)^2\th\tau\bt\tau-(q-r)^2\wt\tau\thb\tau\Big]\Big[(q+r)^2\wh\tau\wb\tau-(q-r)^2\tau\hb\tau\Big]=0.
\end{align}
This is the discrete CKP equation first introduced by Kashaev \cite{Kas96}, and was entitled CKP attributed to Schief \cite{Sch03} who revealed that 
this equation serves the superposition property of the continuous CKP equation. Compared with the discrete AKP and BKP equations, this model is no longer 
in Hirota's bilinear form and instead it takes multiquadratic and quadrilinear form. 

\subsection{Reduction on the measure}

The 3D lattice models, namely the discrete AKP, BKP and CKP equations, are all associated with a nonlocal Riemann--Hilbert problem where the measure depends 
on two spectral parameters and a double integral is involved. The key point in the reduction problem is to reduce the double integral in \eqref{IntEq} to 
a single one connected to a local Riemann--Hilbert problem. The reduction can be performed on the measure $\rd\zeta(k,k')$, see e.g. \cite{ZZN12,FN17}. 
In the cases of the discrete and continuous AKP, BKP and CKP hierarchies, we consider the following reductions, respectively: 
\bse\label{Reduc}
\begin{align}
 &\rd\zeta(k,k')=\sum_{j=1}^{\cN-1}\rd\ld_j(k)\rd k'\delta(k'+\oa_j(k)), \label{Reduc:AKP} \\
 &\rd\zeta(k,k')=\sum_{j=1}^{\cN-1}\Big[\rd\ld_j(k)\rd k'\delta(k'+\oa_j(k))-\rd\ld_j(k')\rd k\delta(k+\oa_j(k'))\Big], \label{Reduc:BKP}\\
 &\rd\zeta(k,k')=\sum_{j=1}^{\cN-1}\Big[\rd\ld_j(k)\rd k'\delta(k'+\oa_j(k))+\rd\ld_j(k')\rd k\delta(k+\oa_j(k'))\Big]. \label{Reduc:CKP}
\end{align}
\ese
The reason why the reductions for BKP and CKP are more complex is because the measures in these cases obey antisymmetry and symmetry properties, respectively, 
cf. \eqref{BKP:Info} and \eqref{CKP:Info}, and such properties must be preserved for the compatibility between the 3D equations and the reduced 2D equations. 

Such a procedure, i.e. reduction on the measure, is effectively restricting spectral parameters acting as free parameters in solutions to a hierarchy of 
3D integrable equations on an algebraic curve, as we can observe this from \eqref{Sol}, or more concretely, from \eqref{AKP:Sol}, \eqref{BKP:Sol} and 
\eqref{CKP:Sol}; in other words, imposing a constraint on the measure is deep down restricting the solution space of a 3D hierarchy, and thus the problem 
turns to be looking for the corresponding reduced hierarchy of 2D equations describing the dynamics of the reduced solution space. This can be realised via 
translating the algebraic curve of the spectral parameters to its corresponding symmetry constraint on the $\tau$-function, and then the reduced 2D hierarchy 
can be obtained by applying the symmetry constraint to the 3D hierarchy, see Section \ref{S:Reduc}. 

\subsection{Multi-dimensional consistency}\label{S:MDC}

The MDC property, namely the property that a discrete equation can be consistently extended to an infinite-parameter family of partial difference equations 
with an infinite number of independent variables, is a consequence of the DL framework. We recall that in the DL for 3D discrete integrable systems, the key 
ingredients are the Cauchy kernel $\Oa_{k,l'}$, and its compatible plane wave factors $\rho_k$ and $\sigma_{k'}$ including an infinite number of discrete 
independent variables; then a hierarchy of 3D discrete equations can be constructed and all the equations are compatible in the sense of possessing a common 
solution admitting an infinite degrees of freedom in terms of spectral parameters (e.g. soliton solution), which in turn implies that the MDC property is 
guaranteed. 

Let us take the Hirota--Miwa equation as an example. One can observe in \eqref{AKP:Info} that the plane wave factors $\rho_k$ and $\sigma_{k'}$ contain 
an infinite number of discrete variables $\{n_\gamma\}$. Thus the soliton solution given by \eqref{AKP:Sol} (i.e. a common solution) solves a hierarchy 
of the Hirota--Miwa equations taking the form of 
\begin{align*}
 &(p_a-p_b)\tau(n_c+1)\tau(n_a+1,n_b+1) \\
 &\quad+(p_b-p_c)\tau(n_a+1)\tau(n_b+1,n_c+1)+(p_c-p_a)\tau(n_b+1)\tau(n_a+1,n_c+1)=0,
\end{align*}
where $p_a$, $p_b$ and $p_c$ are any three arbitrary distinct lattice parameters selected from $\{p_\gamma|\gamma\in\mathbb{Z}^+\}$, and $n_a$, $n_b$ and 
$n_c$ are their corresponding lattice independent variables. Therefore, the lattice independent variables $\{n_\gamma|\gamma\in\mathbb{Z}^+\}$ should be 
understood as the discrete flow variables like the flows variables $\{x_\gamma|\gamma\in\mathbb{Z}^+\}$ in the continuous case. The only difference is that 
in the continuous case $\{x_\gamma|\gamma\in\mathbb{Z}^+\}$ lead to the higher-order symmetries (i.e. the continuous AKP hierarchy) because they play 
different roles (for instance, $x_1$ and $x_2$ are the two preselected variables) in solution structure, while in the discrete case all the lattice variables 
are on the same footing, cf. \eqref{AKP:Info}, and thus the ``higher-order symmetries'' take the same form, and are solved by the common soliton solution -- 
this is the MDC property of the Hirota--Miwa equation. We note that the analysis is also applicable to the more general solution involving the double integral, 
cf. the $\tau$-function defined by the infinite matrices introduced in \eqref{Sol} (see \cite{FN16} for more details), although the MDC property of the 
Hirota--Miwa equation is illustrated by soliton solution here. 

We argue that the MDC property is preserved for the 2D lattice equations arising from a multi-dimensionally consistent 3D lattice equation by our reduction 
approach. This reduction is performed by restricting the spectral parameters in the solutions, which brings a more particular class of solutions to the 3D 
equation. From the view point of the DL method, the hierarchies of 2D equations are obtained by reducing the solution space, and thus the resulting equations 
still have the corresponding reduced common solution, i.e. the 2D equations are multi-dimensionally consistent. 

In addition, sometimes the discrete factors in the plane wave may take different forms, which we call \emph{non-covariant}, leading to the lattice equations 
in the hierarchy also having different forms. In this case, a common solution can still be constructed from the DL framework, namely the non-covariant 
equations in the hierarchy still possess the MDC property. However, this property is not obvious if a lattice equation is studied separately; an example is 
the truncated discrete BSQ equation \eqref{dBSQ:DJM}, which will be discussed later in Section \ref{S:BSQ}. In contrast, we call a discrete plane wave 
\emph{covariant} if all the factors in it take the same form, and this leads to the fact that all the equations in the discrete hierarchy are the same with 
regard to lattice directions and parameters, e.g. the Hirota--Miwa hierarchy discussed above.

\section{Discrete Boussinesq and Korteweg--de Vries equations}\label{S:Comparison}

The reduction problem of a 3D lattice integrable equation amounts to imposing a constraint on the spectral parameters $k$ and $k'$. In this section, we 
illustrate how a generic algebraic curve for the spectral parameters affects the reduction result by concrete examples including the discrete BSQ and KdV 
equations. 

There are two existing discrete BSQ equations, namely a bilinear one and a trilinear one. By performing the reduction associated with the generic algebraic 
curve, the connection between the two equations is clearly revealed. Following this idea, the reduction to the discrete KdV equation is also explained in 
the same way. 

\subsection{Discrete Boussinesq equation: bilinear versus trilinear}\label{S:BSQ}

The discrete BSQ equation arises as the so-called 3-reduction of the discrete AKP equation. In order to understand the structure of the discrete BSQ equation, 
we start with considering the effective plane wave factor in the $\tau$-function of the discrete AKP equation, i.e. 
\begin{align}\label{AKP:PWF}
 \rho_k\sigma_{k'}=\prod_{\gamma=1}^{\infty}\left(\frac{p_\gamma+k}{p_\gamma-k'}\right)^{n_\gamma}
 \exp\left(\sum_{\gamma=1}^{\infty}(k^\gamma-(-k')^\gamma)x_\gamma\right), 
\end{align}
and also the Cauchy kernel $\Oa_{k,k'}=1/(k+k')$. Two different reductions can be imposed on the discrete AKP equation, i.e. the Hirota--Miwa equation, 
and they result in the bilinear and trilinear discrete BSQ equations, respectively. 
\subsubsection{Bilinear discrete Boussinesq equation}
One can impose an algebraic relation on the spectral parameters $k$ and $k'$ as follows: 
\begin{align}\label{dBSQ:DJMReduc}
 R(p_1,p_2,p_3)=\left(\frac{p_1+k}{p_1-k'}\right)\left(\frac{p_2+k}{p_2-k'}\right)\left(\frac{p_3+k}{p_3-k'}\right)=1.
\end{align}
Such a reduction leads to a symmetry constraint on the $\tau$-function defined in \eqref{AKP:Sol}, namely $\thb\tau=\tau$, as one can observe from 
\eqref{AKP:PWF}. The following bilinear 2D equation arises from the Hirota--Miwa equation \eqref{dAKP} as a consequence of imposing the symmetry constraint: 
\begin{align}\label{dBSQ:DJM}
 (p-q)\th\tau\uth\tau+(q-r)\wt\tau\ut\tau+(r-p)\wh\tau\uh\tau=0,
\end{align}
which is the bilinear discrete BSQ equation given by Date, Jimbo and Miwa, see \cite{DJM83}. 

We would like to argue that such a reduction is not sufficient to describe the BSQ structure. This is because the continuous BSQ equation 
$(D_1^4+3D_2^2)\tau\cdot\tau=0$ arises from the continuous AKP equation $(D_1^4-4D_1D_3+3D_2^2)\tau\cdot\tau=0$ by imposing the symmetry constraint 
$\partial_{x_3}\tau=0$, where $x_3$ is the third flow variable in the continuous AKP hierarchy (see \eqref{AKP:PWF}) and $D$ is Hirota's bilinear operator. 
Such a reduction procedure is deep down governed by the algebraic constraint $k^3=(-k')^3$, cf. \eqref{AKP:Info} and \eqref{AKP:Sol}, which implies that 
\eqref{dBSQ:DJMReduc} does not completely reflect the algebraic structure of BSQ. 

For the sake of the consistency between the discrete and continuous BSQ equations, some extra conditions are necessary, 
namely $p_1$, $p_2$ and $p_3$ must be connected in order to reduce \eqref{dBSQ:DJMReduc} to $k^3=(-k')^3$. By expanding \eqref{dBSQ:DJMReduc}, one can 
find that the lattice parameters must be restricted by the relations $p_2=\oa p_1$ and $p_3=\oa^2 p_1$, where $\oa$ is a primitive cube root of unity 
given by $\oa=\exp(2\pi\mathfrak{i}/3)$, cf. \cite{Hie15} (this was also implied in the continuum limit of the equation in \cite{DJM83}). Under this 
assumption the equation \eqref{dBSQ:DJM} turns out to be 
\begin{align}\label{dBSQ:BL}
 \th\tau\uth\tau+\oa\wt\tau\ut\tau+\oa^2\wh\tau\uh\tau=0,
\end{align}
which is a discretisation of the bilinear BSQ equation. This reduction also brings us the reduced effective plane wave factors for the equation given by 
\begin{align}\label{dBSQ:BLPWF}
 \rho_k\sigma_{-\oa^j k}=\prod_{\gamma=2}^{3}\left(\frac{\oa^{\gamma-1}p_1+k}{\oa^{\gamma-1}p_1+\oa^j k}\right)^{n_\gamma}
 \prod_{\gamma\neq2,3,\gamma\in\mathbb{Z}^+}\left(\frac{p_\gamma+k}{p_\gamma+\oa^j k}\right)^{n_\gamma}
 \exp\left(\sum_{\gamma=1}^{\infty}(k^\gamma-(\oa^j k)^\gamma)x_\gamma\right)
\end{align}
for $j=1,2$. This amounts to a constraint on the lattice parameters and hence truncates the dynamics. In other words, \eqref{dBSQ:BL} is a truncated 
equation in the discrete hierarchy. Furthermore, due to the dependence of $p_2$ on $p_1$, the semi-discrete BSQ equation probably cannot be recovered 
from the continuum limit of \eqref{dBSQ:BL} -- the limit immediately leads to the fully continuous one since there is effectively only one lattice 
parameter $p_1$ in the equation. 

Equations \eqref{dBSQ:DJM} and \eqref{dBSQ:BL} are the ones describing the dynamical evolutions with respect to $n_1$ and $n_2$ and both take bilinear forms. 
Similarly, the equations for lattice directions $(n_2,n_3)$ or $(n_1,n_3)$ also take the same forms because the restrictions on $p_1$, $p_2$ and $p_3$ imply 
that the lattice variables $n_1$, $n_2$ and $n_3$ are treated particularly, cf. \eqref{dBSQ:DJMReduc}. A natural question is that what an equation in the 2D 
discrete hierarchy would be if at least one of the discrete variables is selected from $\{n_\gamma|\gamma\geq4\}$ (for instance one may consider the equation 
describing the dynamical evolutions with respect to $n_4$ and $n_5$). The answer is that such equations are in trilinear form and the general structure is 
explained below. 

\subsubsection{Trilinear discrete Boussinesq equation}
We consider the reduction based on a generic cubic curve 
\begin{align}\label{CubicCurve}
 C(-k',k)=(-k')^3-k^3+\alpha_2((-k')^2-k^2)+\alpha_1((-k')-k)=0,
\end{align}
following from \eqref{AlgCurve}, and denote the three roots by $k'=-\oa_j(k)$ for $j=1,2,3$, in which $\oa_1(k)$ and $\oa_2(k)$ are the two primitive roots 
depending on the parameters $\alpha_1$ and $\alpha_2$, and $\oa_3(k)=k$. In the case of $\alpha_1=\alpha_2=0$, the generic cubic curve degenerates to the 
reduction to the BSQ equation $k^3=(-k')^3$ and its solutions are given by $k'=-\oa^jk$. In other words, we have $\oa_j(k)|_{\alpha_1=\alpha_2=0}=\oa^jk$. 
As we have pointed out at the end of Section \ref{S:KP}, imposing the constraint on the spectral parameters is equivalent to considering a subspace of the 
general soliton solution space (i.e. \eqref{AKP:Sol}) of the AKP hierarchy, and thus we expect to find a 2D discrete hierarchy describing the subspace.

The reduction based on the cubic curve provides us with the effective plane wave factors as follows: 
\begin{align}\label{edBSQ:PWF}
 \rho_k\sigma_{-\oa_j(k)}=\prod_{\gamma=1}^{\infty}\left(\frac{p_\gamma+k}{p_\gamma+\oa_j(k)}\right)^{n_\gamma}
 \exp\left(\sum_{\gamma=1}^{\infty}(k^\gamma-(\oa_j(k))^\gamma)x_\gamma\right), \quad j=1,2.
\end{align}
The direct linearisation approach shows in \cite{ZZN12} that such a reduction leads to a hierarchy of extended discrete BSQ equations in trilinear form 
and the one associated with $n_1$ and $n_2$ takes the form of 
\begin{align}\label{edBSQ}
 &(p-q)^2\th\tau\tau\uth\tau-(3q^2-2\alpha_2 q+\alpha_1)\wt\tau\tau\ut\tau-(3p^2-2\alpha_2 p+\alpha_1)\wh\tau\tau\uh\tau \nonumber \\
 &\quad+[(p^2+pq+q^2)-\alpha_2(p+q)+\alpha_1](\ut{\wh\tau}\wt\tau\uh\tau+\uh{\wt\tau}\wh\tau\ut\tau)=0.
\end{align}
For the unextended BSQ equation, we can consider the degeneration $\alpha_1=\alpha_2=0$ and it gives the effective plane wave factors 
\begin{align}\label{dBSQ:PWF}
 \rho_k\sigma_{-\oa^j k}=\prod_{\gamma=1}^{\infty}\left(\frac{p_\gamma+k}{p_\gamma+\oa^j k}\right)^{n_\gamma}
 \exp\left(\sum_{\gamma=1}^{\infty}(k^\gamma-(\oa^j k)^\gamma)x_\gamma\right), \quad j=1,2,
\end{align}
for the unextended fully discrete BSQ hierarchy. The equation in the hierarchy for directions $(n_1,n_2)$ takes the trilinear form of 
\begin{align}\label{dBSQ}
 (p-q)^2\th\tau\tau\uth\tau-3q^2\wt\tau\tau\ut\tau-3p^2\wh\tau\tau\uh\tau
 +(p^2+pq+q^2)(\ut{\wh\tau}\wt\tau\uh\tau+\uh{\wt\tau}\wh\tau\ut\tau)=0.
\end{align}
We note that the nonlinear forms of this trilinear equation \eqref{dBSQ} were even given earlier from the view point of the DL in \cite{NPCQ92,Nij99}. 
From the structure of the plane wave factors we observe that no restriction is imposed on the lattice parameters $\{p_\gamma\}$ and therefore compared to 
the bilinear BSQ equations, the trilinear ones possess more general structure. In fact, the bilinear equations can be recovered from the trilinear equations. 
One can take $\alpha_1=pq+qr+rp$ and $\alpha_2=p+q+r$, then the extended BSQ equation \eqref{edBSQ} turns out to be Equation \eqref{dBSQ:DJM} immediately. 
This is because the cubic curve \eqref{CubicCurve} under such a degeneration becomes \eqref{dBSQ:DJMReduc}. And the equation \eqref{dBSQ:BL} can be recovered 
from the unextended BSQ equation \eqref{dBSQ} by taking $q=\oa p$ -- this can be obtained if one compares their plane wave factors \eqref{dBSQ:PWF} and 
\eqref{dBSQ:BLPWF}. 

The statement can be made as follows: The extended trilinear BSQ equation \eqref{edBSQ} possesses the most general solution structure in the BSQ class. 
Equation \eqref{dBSQ} as an equation without restrictions on the lattice parameters should be referred to as the discrete BSQ equation since it is 
consistent with the continuous BSQ hierarchy in the sense of the curve relation $k^3=(-k')^3$. While Equations \eqref{dBSQ:DJM} and \eqref{dBSQ:BL} taking 
bilinear forms are degenerations of the trilinear equations when extra conditions are imposed on the curve parameters $\alpha_1$ and $\alpha_2$ as well as 
the lattice parameters $p_2$ and $p_3$. Therefore one may think of them as truncations of the extended discrete BSQ equation \eqref{edBSQ} and the discrete 
BSQ equation \eqref{dBSQ} respectively. 

Furthermore, we would like to make a few comments on the MDC property of the above discrete BSQ equations. As is mentioned in Subsection \ref{S:MDC}, the MDC 
property is preserved in dimensional reductions of multi-dimensionally consistent 3D equations. Thus the various BSQ equations, namely \eqref{dBSQ:DJM}, 
\eqref{dBSQ:BL}, \eqref{edBSQ} and \eqref{dBSQ}, all possess the MDC property since the common solutions for the discrete hierarchies follow from their 
respective reductions on that of the discrete AKP hierarchy. In the cases of the non-truncated discrete BSQ equations \eqref{edBSQ} and \eqref{dBSQ}, all 
the discrete BSQ equations in their respective hierarchies take the same form, which is guaranteed by their respective reduced plane wave factors 
\eqref{edBSQ:PWF} and \eqref{dBSQ:PWF}. While in the truncated cases \eqref{dBSQ:DJM} and \eqref{dBSQ:BL}, the equations in each discrete hierarchy may 
look differently. Taking \eqref{dBSQ:DJM} as an example, the equation involving the lattice variables $n_1$ and $n_2$ is bilinear, i.e. \eqref{dBSQ:DJM} itself; 
however, the equation describing the dynamical evolutions with respect to $n_4$ and $n_5$ is trilinear -- we can start from the extended discrete BSQ 
equation for $(n_4,n_5)$ which takes the same form of \eqref{edBSQ} but involves lattice parameters $p_4,p_5$ and extra parameters $\alpha_1,\alpha_2$; 
then the same degeneration from \eqref{edBSQ} to \eqref{dBSQ:DJM}, namely $\alpha_1=pq+qr+rp$ and $\alpha_2=p+q+r$, gives rise to a trilinear equation 
containing parameters $p_4,p_5$ and $p,q,r$.

\subsection{Discrete Korteweg--de Vries equation}\label{S:KdV}
A similar comparison can also be made for the discrete KdV class (although it is not as obvious as the case in the discrete BSQ class). One can introduce 
the 2-reduction on the spectral parameters $k$ and $k'$ as follows: 
\begin{align}\label{dKdV:DJMReduc}
 R(p_2,p_3)=\left(\frac{p_2+k}{p_2-k'}\right)\left(\frac{p_3+k}{p_3-k'}\right)=1.
\end{align}
This is equivalent to the constraint on the $\tau$-function $\hb\tau=\tau$. Making use of the constraint a reduced discrete equation can immediately be 
derived from the Hirota--Miwa equation \eqref{dAKP}, i.e. 
\begin{align}\label{dKdV:DJM}
 (p-q)\th\tau\uh\tau+(q-r)\tau\wt\tau+(r-p)\wt{\uh\tau}\wh\tau=0,
\end{align}
which is the bilinear discrete KdV equation given in \cite{DJM83}. 

Similarly to the case in the discrete BSQ equation, in order to match the discrete KdV equation to its continuous counterpart, one has to take $p_3=-p_2$ 
and therefore the reduction condition becomes the relation $k^2=(-k')^2$, which reduces the continuous KP hierarchy to the continuous KdV hierarchy. 
This gives us the effective plane wave factor that takes the form of 
\begin{align}\label{dKdV:BLPWF}
 \rho_k\sigma_{k}=\left(\frac{-p_2+k}{-p_2-k}\right)^{n_3}
 \prod_{\gamma\neq3,\gamma\in\mathbb{Z}^+}\left(\frac{p_\gamma+k}{p_\gamma-k}\right)^{n_\gamma}
 \exp\left(\sum_{\gamma=1}^{\infty}(k^\gamma-(-k)^\gamma)x_\gamma\right).
\end{align}
Simultaneously this restriction of the lattice parameter $p_3$ on Equation \eqref{dKdV:BL} gives rise to the 6-point bilinear discrete KdV equation 
\begin{align}\label{dKdV:BL}
 (p+q)\uh{\wt\tau}\wh\tau-(p-q)\uh\tau\th\tau=2q\tau\wt\tau.
\end{align}
In the above reduction one can still see that some lattice parameters are restricted and we argue that it leads to some truncated solution structure. 

We now consider the reduction associated with a generic quadratic curve of spectral parameters which is independent of any lattice 
parameters as follows: 
\begin{align}\label{QuadCurve}
 C(-k',k)=(-k')^2-k^2+\alpha_1((-k')-k)=0.
\end{align}
The curve has two roots $k'=-\oa_1(k)=k+\alpha_1$ and $k'=-\oa_2(k)=-k$. Applying such a reduction to \eqref{AKP:PWF} provides us with the effective plane 
wave factor: 
\begin{align}\label{edKdV:PWF}
 \rho_k\sigma_{-\oa_1(k)}=\prod_{\gamma=1}^{\infty}\left(\frac{p_\gamma+k}{p_\gamma+\oa_1(k)}\right)^{n_\gamma}
 \exp\left(\sum_{\gamma=1}^{\infty}(k^\gamma-(\oa_1(k))^\gamma)x_\gamma\right),
\end{align}
where we only choose the primitive root $k'=-\oa_1(k)=k+\alpha_1$. Then the $\tau$-function defined in \eqref{AKP:Sol} under this reduction solves the 
following two compatible equations: 
\begin{align}\label{edKdV}
 (p+q-\alpha_1)\ut{\wh\tau}\wt\tau+(p-q)\ut\tau\th\tau=(2p-\alpha_1)\tau\wh\tau, \quad
 (p+q-\alpha_1)\uh{\wt\tau}\wh\tau-(p-q)\uh\tau\th\tau=(2q-\alpha_1)\tau\wt\tau,
\end{align}
in which either can be thought of as the discrete analogue of the bilinear KdV equation. The derivation of the bilinear discrete KdV equations will be given 
in the next section, and here we note that the equations are actually embedded in the extended discrete BSQ equation \eqref{edBSQ} because the generic cubic 
algebraic curve \eqref{CubicCurve} contains the curve \eqref{QuadCurve}. Equation \eqref{edKdV} is the analogue of the trilinear equation \eqref{edBSQ} on 
the KdV level. 

By taking $\alpha_1=0$ in \eqref{edKdV}, we obtain two compatible equations as follows: 
\begin{align}\label{dKdV}
 (p+q)\ut{\wh\tau}\wt\tau+(p-q)\ut\tau\th\tau=2p\tau\wh\tau, \quad
 (p+q)\uh{\wt\tau}\wh\tau-(p-q)\uh\tau\th\tau=2q\tau\wt\tau.
\end{align}
Considering the compatibility of the two equations, we obtain a third lattice equation taking the form of 
\begin{align}\label{dtToda}
 (p-q)^2\th\tau\uth\tau-(p+q)^2\wt{\uh\tau}\wh{\ut\tau}+4pq\tau^2=0,
\end{align}
which is known as the discrete-time Toda equation. Since the curve \eqref{QuadCurve} degenerates when $\alpha_1=0$, namely $k'=-\oa_j(k)=(-1)^{j-1}k$ for 
$j=1,2$, the equations in \eqref{dKdV} and Equation \eqref{dtToda} are all governed by the effective plane wave factor
\begin{align}\label{dKdV:PWF}
 \rho_k\sigma_{k}=\prod_{\gamma=1}^{\infty}\left(\frac{p_\gamma+k}{p_\gamma-k}\right)^{n_\gamma}
 \exp\left(\sum_{\gamma=1}^{\infty}(k^\gamma-(-k)^\gamma)x_\gamma\right),
\end{align}
where we only choose the primitive root $k'=k$ in order to avoid triviality. This in turn implies that any single equation in \eqref{dKdV} and \eqref{dtToda} 
can be considered as the bilinear discrete KdV equation -- they are just different bilinear forms and the $\tau$-functions are the same. 

The second discrete KdV equation in \eqref{dKdV} looks exactly the same as \eqref{dKdV:BL}. But in fact, the reduced equation from a generic algebraic 
curve possesses more general structure as can be seen if one compares \eqref{dKdV:PWF} and \eqref{dKdV:BLPWF}. More concretely, in the previous case, 
a restriction is made on the lattice parameters, i.e. $p_3=-p_2$, this leads the equation involving the discrete variables $n_2$ and $n_3$ to its 
degeneration $\hb\tau=\tau$, namely the constraint itself. And this is also the reason why the first equation in \eqref{dKdV} is missing in the former case. 
While in the latter case, all the discrete variables are on the same footing and all the discrete KdV equations in the hierarchy are covariant. 

Furthermore, similar to the discrete BSQ case, if we take $\alpha_1=q+r$, the quadratic curve \eqref{QuadCurve} turns out to be the reduction condition 
\eqref{dKdV:DJMReduc}, in other words, the equation \eqref{dKdV:DJM} is a truncated version of the extended discrete KdV equation \eqref{edKdV}. This also 
reflects on the equations themselves, namely we can from \eqref{edKdV} recover the following bilinear discrete 
equations: 
\begin{align}\label{tedKdV}
 (p-r)\ut{\wh\tau}\wt\tau+(p-q)\ut\tau\th\tau=(2p-q-r)\tau\wh\tau, \quad
 (p-r)\uh{\wt\tau}\wh\tau-(p-q)\uh\tau\th\tau=(q-r)\tau\wt\tau,
\end{align}
in which the second equation is exactly the same as \eqref{dKdV:DJM}, while the first one was missing in \cite{DJM83}. 

The MDC property of the discrete KdV equations, similarly to the discrete BSQ equations, is still preserved according to the general theory in Subsection 
\ref{S:MDC}. One remark here is that the equations in \eqref{dKdV} and \eqref{dtToda} should be understood as members in the same discrete KdV hierarchy. 

\bigskip

In the above analysis of the discrete BSQ and KdV equations, we considered the reduction associated with $\alpha_\cN=1$. The results show that the reduced 
2D lattice equations have more general structure, and the bilinear discrete BSQ and KdV equations given in the literature are the truncations of the obtained 
2D integrable lattice equations by taking particular values for the parameters $\alpha_i$. 

However, one can observe that $\alpha_1$ and $\alpha_2$ can be removed from Equations \eqref{edKdV} and \eqref{edBSQ}, respectively, after some simple 
transforms on the lattice parameters $p$ and $q$. For instance, the transforms $p\rightarrow p+\alpha_1/2$ and $q\rightarrow q+\alpha_1/2$ 
immediately bring \eqref{edKdV} to \eqref{dKdV}. Similar transforms can also eliminate $\alpha_2$ in \eqref{edBSQ}. This is because the term of the second 
highest degree, i.e. $\alpha_{\cN-1}(t^{\cN-1}-k^{\cN-1})$, in the algebraic curve \eqref{AlgCurve} can be absorbed by certain transforms on $t$ and $k$; 
in other words, $\alpha_{\cN-1}$ is redundant in the reductions. Thus, without losing generality, for dimensional reductions of the discrete KP equations, 
we only need to consider a generic algebraic curve of degree $\cN$ in the form of 
\begin{align}\label{AlgCurv:A}
 C_\cN(t,k)\doteq P_\cN(t)-P_\cN(k)=0, \quad \hbox{where} \quad P_\cN(t)=t^\cN+\sum_{j=1}^{\cN-2}\alpha_{j}t^j, \quad \cN=2,3,\cdots.
\end{align}
For convenience, we still denote the $\cN$ roots to the algebraic equation by $t=\oa_j(k)$ for $j=1,2,\cdots,\cN$ in which $\oa_j$ satisfies 
$\oa_j(k)|_{\alpha_1=\cdots=\alpha_{\cN-2}=0}=\oa^jk$ with $\oa=\exp(2\pi\mathfrak{i}/\cN)$. 

Considering the algebraic curve $C_\cN(-k',k)=0$, we find that in the case of $\cN=2$, the curve takes the form of $C_2(-k,k)=(-k')^2-k^2=0$. Thus there 
is no extended form for the discrete KdV equation, and every single equation in \eqref{dKdV} and \eqref{dtToda} can be thought of as its standard form. 
When $\cN=3$, the algebraic curve $C_3(-k',k)=(-k')^3-k^3+\alpha_1((-k')-k)=0$ contains only one parameter $\alpha_1$. Hence the extended BSQ equation 
can be obtained immediately by taking $\alpha_2=0$ in \eqref{edBSQ}, which takes the form of 
\begin{align}\label{edBSQ2}
 (p-q)^2\th\tau\tau\uth\tau-(3q^2+\alpha_1)\wt\tau\tau\ut\tau-(3p^2+\alpha_1)\wh\tau\tau\uh\tau
 +(p^2+pq+q^2+\alpha_1)(\ut{\wh\tau}\wt\tau\uh\tau+\uh{\wt\tau}\wh\tau\ut\tau)=0.
\end{align}
The degeneration case $\alpha_1=0$ then gives us the standard trilinear BSQ equation \eqref{dBSQ}.

\section{General reduction scheme}\label{S:Reduc}
In this section, we would like to consider the reduction associated with a generic algebraic curve involving extended parameters $\alpha_j$ for the 
discrete AKP, BKP and CKP equations. The aim is to give a unified expression of the reduced 2D lattice equations in each case. Examples include some 
new integrable lattice equations such as the discrete SK, KK and HS equations. 

As we discussed in the previous section, $\alpha_{\cN-1}$ is redundant in the reduction and it can be removed by some suitable transforms on the 
spectral parameters. Therefore, for reductions of the discrete AKP equation, we only consider constraint $C_\cN(-k,k)=0$ defined by \eqref{AlgCurv:A}, 
and it is corresponding to the reduction condition \eqref{Reduc:AKP}. However, since the measure reductions for BKP and CKP, i.e. \eqref{Reduc:BKP} and 
\eqref{Reduc:CKP}, must be anti-symmetrised and symmetrised, respectively, due to their original properties on the 3D level, cf. \eqref{BKP:Sol} and 
\eqref{CKP:Info}, in both cases the generic curve must take the following form: 
\begin{align}\label{AlgCurv:BC}
 C_\cN(t,k)\doteq P_\cN(t)-P_\cN(k)=0 \quad \hbox{with} \quad
 P_\cN(t)=
 \left\{
 \begin{array}{ll}
  t^\cN+\sum_{i=1}^{j-1}\alpha_{2i}t^{2i}, & \cN=2j, \\
  t^\cN+\sum_{i=0}^{j-1}\alpha_{2i+1}t^{2i+1}, & \cN=2j+1 
 \end{array}
 \right.
\end{align}
for $j=1,2,\cdots$, and we still denote all the roots of the curve by $t=\oa_j(k)$. This curve immediately leads to the fact that 
\begin{align}\label{AlgCurv:OddEven}
 \quad C_\cN(-k',k)=0 \quad \Longleftrightarrow \quad C_\cN(-k,k')=0,
\end{align}
because $P(t)$ in the curve \eqref{AlgCurv:BC} is either an odd degree polynomial or an even degree polynomial. The reduction $C(-k,k')=0$ for the discrete 
AKP, BKP and CKP equations can be translated to a symmetry constraint and we give it in the following proposition: 

\begin{proposition}
The reduction condition $C_\cN(-k',k)=0$ determined by \eqref{AlgCurv:A} (resp. \eqref{AlgCurv:BC}) for the discrete AKP (resp. BKP and CKP) equation(s) 
provides the symmetry constraint 
\begin{align}\label{Reduc:tau}
 \left(\prod_{j=1}^\cN\rT_{-\oa_j(-p_\gamma)}\right)\tau\Big|_{C_\cN(-k',k)=0}=\tau \quad \hbox{for} \quad \gamma=1,2,\cdots.
\end{align}
\end{proposition}
\begin{proof}
As is mentioned in Subsection \ref{S:DL} the $\tau$-function in the framework of the DL is defined via $\det(1+\bOa\cdot\bC)$ and the dynamical variables 
only appear in the effective plane wave factor $\rho_k\sigma_{k'}$ in the infinite matrix $\bC$, cf. \eqref{Sol}. Thus, we only need to prove in each 
case that 
\begin{align}\label{Reduc:k}
\left(\prod_{j=1}^\cN\rT_{-\oa_j(-p_\gamma)}\right)\rho_k\sigma_{k'}\Big|_{C_\cN(-k',k)=0}=\rho_k\sigma_{k'} \quad \hbox{for} \quad \gamma=1,2,\cdots.
\end{align}
In the case of AKP, namely the plane wave factors given by \eqref{AKP:Sol}, we can calculate that 
\begin{align*}
 \left(\prod_{j=1}^\cN\rT_{-\oa_j(-p_\gamma)}\right)\rho_k\sigma_{k'}=
 \prod_{j=1}^\cN\left(\frac{-\oa_j(-p_\gamma)+k}{-\oa_j(-p_\gamma)-k'}\right)\rho_k\sigma_{k'}=
 \frac{P_\cN(k)-P_\cN(-p)}{P_\cN(-k')-P_\cN(-p)}\rho_k\sigma_{k'}=\rho_k\sigma_{k'},
\end{align*}
where in the last step we used the relation $P_\cN(k)=P_\cN(-k')$, i.e.  the constraint $C_\cN(-k',k)=0$ governed by the curve \eqref{AlgCurv:A}. Likewise, 
in the cases of BKP and CKP, we can from the plane wave factors given by \eqref{BKP:Info} and \eqref{CKP:Info} obtain that  
\begin{align*}
 \left(\prod_{j=1}^\cN\rT_{-\oa_j(-p_\gamma)}\right)\rho_k\sigma_{k'}&=
 \prod_{j=1}^\cN\left(\frac{-\oa_j(-p_\gamma)+k}{-\oa_j(-p_\gamma)-k}\right)\left(\frac{-\oa_j(-p_\gamma)+k'}{-\oa_j(-p_\gamma)-k'}\right)\rho_k\sigma_{k'} \\
 &=\left(\frac{P_\cN(k)-P_\cN(-p)}{P_\cN(-k)-P_\cN(-p)}\right)\left(\frac{P_\cN(k')-P_\cN(-p)}{P_\cN(-k')-P_\cN(-p)}\right)\rho_k\sigma_{k'}=\rho_k\sigma_{k'},
\end{align*}
where the last equality holds because of the constraints given in \eqref{AlgCurv:OddEven}. Thus, in any case Equation \eqref{Reduc:k} is proven and
consequently the symmetry constraint \eqref{Reduc:tau} is obtained. 
\end{proof}

This is the generic reduction on the bilinear/multilinear forms of the discrete KP equations and in the following we show that under such a reduction, 
coupled systems of 2D lattice equations can be obtained. In practice, we define some new variables $\sigma_{i}(p_\gamma)$ as 
\begin{align}\label{sigma}
 \sigma_i(p_\gamma)=\left(\prod_{j=1}^i\rT_{-\oa_j(-p_\gamma)}\right)\tau, \quad i=0,1,2,\cdots,\cN-1.
\end{align}
Following from the reduction \eqref{Reduc:tau} and the definition \eqref{sigma}, one can easily prove that the new variables obey the relations as follows:
\begin{align}\label{sigma:Prop}
 \sigma_0(p_\gamma)=\tau, \quad \sigma_{i}(p_\gamma)=\rT_{-\oa_i(-p_\gamma)}\sigma_{i-1}(p_\gamma), \quad \sigma_{\cN-1}(p_\gamma)=\rT_{p_\gamma}^{-1}\tau.
\end{align}
These variables appear in the reduced 2D lattice equations, namely they are the components in the obtained coupled systems of discrete equations. 

We can now consider a generic 3D discrete integrable equation that takes the form of 
\begin{align}
 \cF_{p,q,r}(\tau,\wt\tau,\wh\tau,\wb\tau,\th\tau,\hb\tau,\bt\tau,\thb\tau)
 =\cF_{p,q,r}(\tau,\wt\tau,\wh\tau,\rT_r\tau,\th\tau,\rT_r\wh\tau,\rT_r\wt\tau,\rT_r\th\tau)=0,
\end{align}
where we denote $p=p_1$, $q=p_2$ and $r=p_3$ according to the convention that we have claimed in Section \ref{S:KP}. Without losing generality, we take 
$r=-\oa_i(-p)$ for $i=1,2,\cdots,\cN-1$ respectively, and as a result we obtain a coupled system of $\cN-1$ discrete equations as follows: 
\begin{align}
 \cF_{p,q,-\oa_i(-p)}(\sigma_{i-1}(p),\wt\sigma_{i-1}(p),\wh\sigma_{i-1}(p),\sigma_{i}(p),
 \th\sigma_{i-1}(p),\wh\sigma_{i}(p),\wt\sigma_{i}(p),\th\sigma_{i}(p))=0, \quad i=1,2,\cdots,\cN-1,
\end{align}
where $\sigma_0(p)=\tau$ and $\sigma_{\cN-1}(p)=\ut\tau$ according to \eqref{sigma:Prop}. This coupled system has effectively $\cN-1$ dependent variables 
$\tau,\sigma_1(p),\cdots\sigma_{\cN-2}(p)$. Likewise, one can also take $r=-\oa_i(-q)$ for $i=1,2,\cdots,\cN-1$ respectively and obtain a coupled system 
in the same form but in terms of the lattice parameter $q$. 

\subsection{Reductions of the discrete AKP equation}\label{S:AKPReduc}
The discrete AKP equation, i.e. the Hirota--Miwa equation, takes the form of \eqref{dAKP}. We can therefore write down its $\cN$-reduction as follows: 
\begin{align}\label{Reduc:dAKP}
 (p-q)\sigma_i\th\sigma_{i-1}+(q+\oa_i(-p))\wt\sigma_{i-1}\wh\sigma_{i}-(\oa_i(-p)+p)\wh\sigma_{i-1}\wt\sigma_{i}=0, \quad i=1,2,\cdots,\cN-1,
\end{align}
where we denote $\sigma_i=\sigma_i(p)$ and in addition $\sigma_{0}(p)=\tau$ and $\sigma_{\cN-1}(p)=\ut\tau$. The coupled system is integrable (similarly 
for the reductions of BKP and CKP) in the sense of the MDC property, or alternatively, in the sense of the existence of a common explicit solution solving 
the whole hierarchy, e.g. soliton solution given below, which is a natural consequence of the DL framework discussed in Subsection \ref{S:MDC}. We also note 
that a coupled system in terms of $q$ can also be obtained similarly and it is compatible with \eqref{Reduc:dAKP} due to the MDC property. 

The $N$-soliton solution of the coupled system (also of the other equations in the hierarchy) is given by $\tau=\det(1+\bA\bM)$. Here the constant matrix 
$\bA$ and the Cauchy matrix $\bM$ are given by (cf. \cite{FN17} as the structure is the same as that in the continuous case, similarly for the results below 
for the reductions of the discrete BKP and CKP equations) 
\begin{align}\label{A:Sol}
 &\bA=\diag(A_{1,1},\cdots,A_{1,N_1};\cdots;A_{j,1},\cdots,A_{j,N_j};\cdots;A_{\cN-1,1},\cdots,A_{\cN-1,N_{\cN-1}}), \nonumber \\
 &\bM=(M_{(j,j'),(i,i')})_{j,i=1,\cdots,\cN-1,j'=1,\cdots,N_j,i'=1,\cdots,N_i}, \quad 
 M_{(j,j'),(i,i')}=\frac{\rho_{k_{i,i'}}\sigma_{-\oa_j(k_{j,j'})}}{k_{i,i'}-\oa_j(k_{j,j'})},
\end{align}
where the plane wave factors take the form of 
\begin{align*}
 &\rho_k=\prod_{\gamma=1}^{\infty}\left[(p_\gamma+k)^{n_\gamma}\prod_{j=1}^{\cN-1}(-\oa_j(-p_\gamma)+k)^{n_\gamma^{(j)}}\right]
 \exp\left(\sum_{\gamma=1}^{\infty}k^\gamma x_\gamma\right), \\
 &\sigma_{k'}=\prod_{\gamma=1}^{\infty}\left[(p_\gamma-k')^{-n_\gamma}\prod_{j=1}^{\cN-1}(-\oa_j(-p_\gamma)-k')^{-n_\gamma^{(j)}}\right]
 \exp\left(\sum_{\gamma=1}^{\infty}(-k')^\gamma x_\gamma\right).
\end{align*}
As long as the structure of the $\tau$-function is given, the $\sigma_i$ can be easily obtained by acting the discrete shifts on the $\tau$-function 
according to the definition \eqref{sigma}. In the following, we give the nontrivial examples for $\cN=2$ and $\cN=3$. 

\paragraph{Discrete Korteweg--de Vries equation ($\cN=2$).}
The case of $\cN=2$ only gives the curve $C_2(-k,k)=(-k')^2-k^2=0$. In this case we take $r=-\oa_1(-p)=-p$ and obtain the first equation in \eqref{dKdV}. 
Alternatively, we can take $r=-\oa_1(-q)=-q$ and then general formula \eqref{Reduc:dAKP} gives rise to the second equation in \eqref{dKdV}. The two equations 
are compatible in the sense of the MDC property, i.e. the common $\tau$-function given by \eqref{A:Sol}. 

\paragraph{Discrete Boussinesq equation ($\cN=3$).}
When $\cN=3$, according to the the general scheme, the variable $\sigma_1$ is needed and the coupled system form arises. We can then take $r=-\oa_1(-p)$ 
and $r=-\oa_2(-p)$ and obtain the coupled system for $\tau$ and $\sigma\doteq\sigma_1(p)$: 
\bse\label{dBSQ:coupled}
\begin{align}
 &(p-q)\sigma\th\tau+(q+\oa_1(-p))\wt\tau\wh\sigma-(\oa_1(-p)+p)\wh\tau\wt\sigma=0, \\
 &(p-q)\ut\tau\th\sigma+(q+\oa_2(-p))\wt\sigma\ut{\wh\tau}-(\oa_2(-p)+p)\wh\sigma\tau=0.
\end{align}
\ese
We refer to this system as the extended discrete BSQ equation. Alternatively, we can also derive a coupled system in a similar form for $\tau$ and 
$\sigma_1(q)$, which is compatible with \eqref{dBSQ:coupled}. 

We also note that eliminating $\sigma$ in \eqref{dBSQ:coupled} gives us the trilinear extended discrete BSQ equation \eqref{edBSQ2}; the $\tau$-function 
for the BSQ equation solves the coupled system \eqref{dBSQ:coupled} and the trilinear equation \eqref{edBSQ2} simultaneously. The degeneration $\alpha_1=0$ 
gives the unextended version of Equation \eqref{dBSQ:coupled} and it is corresponding to \eqref{dBSQ}.

\subsection{Reductions of the discrete BKP equation}\label{S:BKPReduc}
The $\cN$-reduction of the discrete BKP equation \eqref{dBKP} following from the general scheme gives us the following coupled system 
(if one takes $r=-\oa_i(-p)$ for $i=1,2,\cdots,\cN-1$ without loss of generality): 
\begin{align}\label{Reduc:dBKP}
 &(p-q)(q+\oa_i(-p))(\oa_i(-p)+p)\sigma_{i-1}\th\sigma_i-(p+q)(p-\oa_i(-p))(q+\oa_i(-p))\wt\sigma_{i-1}\wh\sigma_i \nonumber \\
 &\quad-(\oa_i(-p)-p)(\oa_i(-p)-q)(p-q)\sigma_i\th\sigma_{i-1}+(q-\oa_i(-p))(q+p)(\oa_i(-p)+p)\wh\sigma_{i-1}\wt\sigma_{i}=0
\end{align}
for $i=1,2,\cdots,\cN-1$, where we denote $\sigma_i\doteq\sigma_i(p)$ for convenience and $\sigma_0(p)=\tau$ and $\sigma_{\cN-1}=\ut\tau$. 
Alternatively, one can also take $r=-\oa_i(-q)$ for $i=1,2,\cdots,\cN-1$ and obtain a similar coupled system associated with the lattice parameter $q$, 
which is compatible with \eqref{Reduc:dBKP} due to the MDC property. 

The coupled system has the $N$-soliton solution $\tau$ determined by 
\begin{align}
 \tau^2=\det\Bigg[1+
 \left(
 \begin{array}{cc}
  \bA & 0 \\
  0 & -\bA
 \end{array}
 \right)
 \left(
 \begin{array}{cc}
  \bM & 0 \\
  0 & \bM'
 \end{array}
 \right)
 \Bigg],
\end{align}
where the matrix $\bA$ is given by 
\begin{align*}
 \bA=\diag(A_{1,1},\cdots,A_{1,N_1};\cdots;A_{j,1},\cdots,A_{j,N_j};\cdots;A_{\cN-1,1},\cdots,A_{\cN-1,N_{\cN-1}}),
\end{align*}
and the Cauchy matrices $\bM$ and $\bM'$ take the form of 
\bse\label{B:M}
\begin{align}
 &\bM=(M_{(j,j'),(i,i')})_{j,i=1,\cdots,\cN-1,j'=1,\cdots,N_j,i'=1,\cdots,N_i}, \, 
 M_{(j,j'),(i,i')}=\rho_{k_{i,i'}}\frac{1}{2}\frac{k_{i,i'}+\oa_j(k_{j,j'})}{k_{i,i'}-\oa_j(k_{j,j'})}\rho_{-\oa_j(k_{j,j'})}, \\
 &\bM'=(M_{(j,j'),(i,i')}')_{j,i=1,\cdots,\cN-1,j'=1,\cdots,N_j,i'=1,\cdots,N_i}, \,
 M_{(j,j'),(i,i')}'=\rho_{-\oa_i(k_{i,i'})}\frac{1}{2}\frac{\oa_i(k_{i,i'})+k_{j,j'}}{\oa_i(k_{i,i'})-k_{j,j'}}\rho_{k_{j,j'}},
\end{align}
\ese
respectively. Here the plane wave factor $\rho_k$ is given by 
\begin{align*}
 \rho_k=\prod_{\gamma=1}^{\infty}\left[\left(\frac{p_\gamma+k}{p_\gamma-k}\right)^{n_\gamma}\prod_{j=1}^{\cN-1}\left(\frac{-\oa_j(-p_\gamma)+k}{-\oa_j(-p_\gamma)-k}\right)^{n_\gamma^{(j)}}\right]
 \exp\left(\sum_{\gamma=0}^{\infty}k^{2\gamma+1}x_{2\gamma+1}\right).
\end{align*}
Below we give an example for $\cN=3$, which is the discrete SK equation in extended form. 

\paragraph{Discrete Sawada--Kotera equation ($\cN=3$).}
The coupled system that describes the structure of the extended discrete SK equation according to the above scheme is given by 
\bse\label{dSK:coupled}
\begin{align}
 &(p-q)(q+\oa_1(-p))(\oa_1(-p)+p)\tau\th\sigma-(p+q)(p-\oa_1(-p))(q+\oa_1(-p))\wt\tau\wh\sigma \nonumber \\
 &\quad-(\oa_1(-p)-p)(\oa_1(-p)-q)(p-q)\sigma\th\tau+(q-\oa_1(-p))(q+p)(\oa_1(-p)+p)\wh\tau\wt\sigma=0, \\
 &(p-q)(q+\oa_2(-p))(\oa_2(-p)+p)\sigma\wh\tau-(p+q)(p-\oa_2(-p))(q+\oa_2(-p))\wt\sigma\ut{\wh\tau} \nonumber \\
 &\quad-(\oa_2(-p)-p)(\oa_2(-p)-q)(p-q)\ut\tau\th\sigma+(q-\oa_2(-p))(q+p)(\oa_2(-p)-p)\wh\sigma\tau=0,
\end{align}
\ese
where we take $r=-\oa_1(-p)$ and $r=-\oa_2(-p)$ respectively in the discrete BKP equation and denote $\sigma\doteq\sigma_1(p)$. 
Similarly one can also derive a compatible coupled system for $\tau$ 
and $\sigma_1(q)$. It is not yet clear whether or not one can eliminate the $\sigma$-function from the coupled systems and express the extended discrete SK 
equation by only the $\tau$-function as a scalar multilinear equation. The degeneration $\alpha_1=\alpha_2=0$ is corresponding to the unextended version 
of the coupled system \eqref{dSK:coupled}, which should be referred to as the discrete SK equation. 

We also note that a different reduction $\thb\tau=\tau$ was considered in \cite{HZ13} (see also \cite{TH96} for a slightly different reduction) 
for the discrete BKP equation \eqref{dBKP} resulting in the following bilinear equation: 
\begin{align}
 &(p-q)(q-r)(r-p)\tau^2+(p+q)(p+r)(q-r)\wt\tau\ut\tau \nonumber \\
 &\quad+(r+p)(r+q)(p-q)\uth\tau\th\tau+(q+r)(q+p)(r-p)\wh\tau\uh\tau=0.
\end{align}
However, the algebraic curve for the spectral parameters $k$ and $k'$ behind such a reduction is 
\begin{align}\label{MultiQuad}
 \left(\frac{p_1+k}{p_1-k}\right)\left(\frac{p_2+k}{p_2-k}\right)\left(\frac{p_3+k}{p_3-k}\right)
 \left(\frac{p_1+k'}{p_1-k'}\right)\left(\frac{p_2+k'}{p_2-k'}\right)\left(\frac{p_3+k'}{p_3-k'}\right)=1,
\end{align}
which is effectively an algebraic curve biquadratic in $k$ and $k'$ and it does not match \eqref{AlgCurve}. This implies that the obtained bilinear equation 
does not possess the structure of the SK equation, but it is still an integrable equation having different nonlinear behaviour subject to the algebraic curve 
\eqref{MultiQuad}. Differently, in the BSQ case the algebraic curve associated with the reduction $\thb\tau=\tau$ is a cubic curve, i.e. \eqref{dBSQ:DJMReduc}, 
and hence the reduced equation \eqref{dBSQ:DJM} can still be thought of a discretisation of the BSQ equation, although it is a truncation.

\subsection{Reductions of the discrete CKP equation}\label{S:CKPReduc}
Finally, we consider the $\cN$-reduction of the discrete CKP equation. If one takes $r=-\oa_i(-p)$ for $i=1,2,\cdots,\cN-1$ (or alternatively $r=-\oa_i(-q)$ 
is also allowed), the general reduction scheme in the case of the discrete CKP provides us with the following coupled system of lattice equations: 
\begin{align}\label{Reduc:dCKP}
 &\Big[(p-q)^2(q+\oa_i(-p))^2(\oa_i(-p)+p)^2\sigma_{i-1}\th\sigma_{i}+(p+q)^2(p-\oa_i(-p))^2(q+\oa_i(-p))^2\wt\sigma_{i-1}\wh\sigma_{i} \nonumber \\
 &\quad-(q-\oa_i(-p))^2(q+p)^2(\oa_i(-p)+p)^2\wh\sigma_{i-1}\wt\sigma_{i}-(\oa_i(-p)-p)^2(\oa_i(-p)-q)^2(p-q)^2\sigma_{i}\th\sigma_{i-1}\Big]^2 \nonumber \\
 &\qquad-4(p^2-q^2)^2(p^2-(\oa_i(-p))^2)^2\Big[(q-\oa_i(-p))^2\th\sigma_{i-1}\wt\sigma_{i}-(q+\oa_i(-p))^2\wt\sigma_{i-1}\th\sigma_{i}\Big] \nonumber \\
 &\quad\qquad\times\Big[(q-\oa_i(-p))^2\wh\sigma_{i-1}\sigma_{i}-(q+\oa_i(-p))^2\sigma_{i-1}\wh\sigma_{i}\Big]=0 \quad 
 \hbox{for} \quad i=1,2,\cdots,\cN-1,
\end{align}
where we still have from \eqref{sigma:Prop} that $\sigma_0(p)=\tau$ and $\sigma_{\cN-1}(p)=\ut\tau$ and denote $\sigma_i\doteq\sigma_i(p)$. And again we 
note that the coupled system for $\tau$ and $\sigma_i(q)$ is compatible with \eqref{Reduc:dCKP}. 

The soliton-type solution to the coupled system \eqref{Reduc:dCKP} (also to its hierarchy) takes the form of 
\begin{align}\label{C:BLSoliton}
 \tau=\det\Bigg[1+
 \left(
 \begin{array}{cc}
  \bA & 0 \\
  0 & \bA
 \end{array}
 \right)
 \left(
 \begin{array}{cc}
  \bM & 0 \\
  0 & \bM'
 \end{array}
 \right)
 \Bigg]
\end{align}
with the same matrix $\bA$ in Subsection \ref{S:BKPReduc} and the Cauchy matrices $\bM$ and $\bM'$ given by 
\bse\label{C:M}
\begin{align}
 &\bM=(M_{(j,j'),(i,i')})_{j,i=1,\cdots,\cN-1,j'=1,\cdots,N_j,i'=1,\cdots,N_i}, \quad 
 M_{(j,j'),(i,i')}=\frac{\rho_{k_{i,i'}}\rho_{-\oa_j(k_{j,j'})}}{k_{i,i'}-\oa_j(k_{j,j'})}, \\
 &\bM'=(M_{(j,j'),(i,i')})_{j,i=1,\cdots,\cN-1,j'=1,\cdots,N_j,i'=1,\cdots,N_i}, \quad 
 M_{(j,j'),(i,i')}'=\frac{\rho_{-\oa_i(k_{i,i'})}\rho_{k_{j,j'}}}{-\oa_i(k_{i,i'})+k_{j,j'}},
\end{align}
\ese
where $\rho_k$ is also exactly the same as the one given in Subsection \ref{S:BKPReduc}. The cases when $\cN=3$ and $\cN=4$ would be quite interesting in 
this class since they give rise to discretisations of the KK and HS equations in extended form and such results (even the unextended cases) have not 
yet be given elsewhere, to the best of the authors' knowledge. 

\paragraph{Discrete Kaup--Kupershmidt equation ($\cN=3$).}
The discrete KK equation can be obtained as a two-component system for $\tau$ and $\sigma\doteq\sigma_{1}(p)$ if one takes $r=-\oa_1(-p)$ and $r=-\oa_2(-p)$ 
respectively in the discrete CKP equation. The extended discrete KK equation takes the form of
\bse\label{dKK:coupled}
\begin{align}
 &\Big[(p-q)^2(q+\oa_1(-p))^2(\oa_1(-p)+p)^2\tau\th\sigma+(p+q)^2(p-\oa_1(-p))^2(q+\oa_1(-p))^2\wt\tau\wh\sigma \nonumber \\
 &\quad-(q-\oa_1(-p))^2(q+p)^2(\oa_1(-p)+p)^2\wh\tau\wt\sigma-(\oa_1(-p)-p)^2(\oa_1(-p)-q)^2(p-q)^2\sigma\th\tau\Big]^2 \nonumber \\
 &\qquad-4(p^2-q^2)^2(p^2-(\oa_1(-p))^2)^2\Big[(q-\oa_1(-p))^2\th\tau\wt\sigma-(q+\oa_1(-p))^2\wt\tau\th\sigma\Big] \nonumber \\
 &\quad\qquad\times\Big[(q-\oa_1(-p))^2\wh\tau\sigma-(q+\oa_1(-p))^2\tau\wh\sigma\Big]=0, \\
 &\Big[(p-q)^2(q+\oa_2(-p))^2(\oa_2(-p)+p)^2\sigma\wh\tau+(p+q)^2(p-\oa_2(-p))^2(q+\oa_2(-p))^2\wt\sigma\ut{\wh\tau} \nonumber \\
 &\quad-(q-\oa_2(-p))^2(q+p)^2(\oa_2(-p)+p)^2\wh\sigma\tau-(\oa_2(-p)-p)^2(\oa_2(-p)-q)^2(p-q)^2\ut\tau\th\sigma\Big]^2 \nonumber \\
 &\qquad-4(p^2-q^2)^2(p^2-(\oa_2(-p))^2)^2\Big[(q-\oa_2(-p))^2\th\sigma\tau-(q+\oa_2(-p))^2\wt\sigma\wh\tau\Big] \nonumber \\
 &\quad\qquad\times\Big[(q-\oa_2(-p))^2\wh\sigma\ut\tau-(q+\oa_2(-p))^2\sigma\ut{\wh\tau}\Big]=0.
\end{align}
\ese
Alternatively, a compatible system for $\tau$ and $\sigma_1(q)$ can also be derived if one takes $r=-\oa_1(-q)$ and $r=-\oa_2(-q)$ respectively 
in the discrete CKP equation. The unextended KK equation is a particular case of \eqref{dKK:coupled} when $\alpha_1=\alpha_2=0$. 

Similarly to the discrete SK equation, the reduction $\thb\tau=\tau$ together with the discrete CKP equation \eqref{dCKP} gives rise to a very simple 
scalar quadrilinear discrete equation 
\begin{align}
 &\Big[(p-q)^2(q-r)^2(r-p)^2\tau^2+(p+q)^2(p+r)^2(q-r)^2\wt\tau\ut\tau \nonumber \\
 &\quad-(q+r)^2(q+p)^2(r-p)^2\wh\tau\uh\tau-(r+p)^2(r+q)^2(p-q)^2\uth\tau\th\tau\Big]^2 \nonumber \\
 &\qquad-4(p^2-q^2)^2(p^2-r^2)^2\Big[(q+r)^2\th\tau\uh\tau-(q-r)^2\wt\tau\tau\Big]\Big[(q+r)^2\wh\tau\uth\tau-(q-r)^2\tau\ut\tau\Big]=0.
\end{align}
But its solution structure does not reflect the structure of the KK equation since the reduction is still associated with the biquadratic algebraic 
curve \eqref{MultiQuad} rather than the algebraic curve $C_3(-k,k)=0$ determined by \eqref{AlgCurv:BC}.

\paragraph{Discrete Hirota--Satsuma equation ($\cN=4$).} The discrete HS equation is a three-component system for $\tau$, $\sigma_{1}(p)$ and $\sigma_{2}(p)$ 
(or consistently a three-component system for $\tau$, $\sigma_1=\sigma_{1}(q)$ and $\sigma_2=\sigma_{2}(q)$ in terms of the lattice direction associated with 
the parameter $q$). In practice, one can take $r=-\oa_i(-p)$ for $i=1,2,3$ respectively in the discrete CKP without losing generality. We omit the explicit 
formulae here and one can refer to the case of $\cN=4$ in \eqref{Reduc:dCKP} , which is the extended HS equation. And the HS equation is the case when 
$\alpha_1=\alpha_2=0$, i.e. the unextended case. 

\bigskip

Finally, we would like to point out that the coupled systems we have listed in this section are the integrable discretisations of their corresponding 
continuous hierarchies in the sense that the respective continuum limits of the $\tau$-functions lead to the continuous ones given in \cite{FN17}. 
The limit procedure on the discrete equations might be slightly subtle -- one can calculate the continuum limits of the discrete equations first and 
then eliminate the variables $\sigma_i$. 

\section{Discussions}\label{S:Concl}
Reductions of the discrete AKP, BKP and CKP equations are considered in a unified way. For lower-rank examples, the coupled systems could be reformulated 
into scalar equations which only involve the $\tau$-function, for instance, the discrete KdV and BSQ equations can be written as the bilinear equation 
\eqref{dKdV} and the trilinear equation \eqref{dBSQ} respectively. However, this is highly nontrivial, especially in the coupled systems as the reductions 
of the discrete CKP equation due to the complexity of its quadrilinear form. We believe that the coupled system form is a proper way to express the 2D 
reduced discrete equations since a unified expression can be written down explicitly. 

We only considered the reductions on the bilinear/multilinear form of the discrete AKP, BKP and CKP equations since the $\tau$-function is 
the best candidate to describe the solution structure (namely, the effective plane wave factor and the Cauchy kernel) of an integrable equation/hierarchy. 
An interesting question would be looking for the nonlinear forms of the obtained 2D lattice integrable equations. This can be done via the DL framework and 
successful examples include the discrete KdV equation \cite{HJN16} and the discrete BSQ equation \cite{ZZN12}. However, the nonlinear forms for the 2D lattice 
equations arising from the reductions of the discrete BKP and CKP equations are not yet clear. In fact, on the continuous level there are some remarkable 
transforms between the SK and KK equations \cite{FG80a,FG80b} (see also \cite{FN17}). Nevertheless, no such result has been found on the discrete level 
(even not in the semi-discrete case according to \cite{AP11}). In addition, there exists a trilinear discrete Tzitzeica equation \cite{Sch99} (see also 
\cite{Adl11}). It is also not clear if this equation is related to the discrete SK equation \eqref{dSK:coupled}. 

Finally, it was pointed out in \cite{OST10} that the famous pentagram map is deeply related to the discrete BSQ equation in the sense that the pentagram map 
leads to the continuous BSQ equation in continuum limit. Recently, it was shown by Hietarinta and Maruno \cite{Hie15} that the pentagram map as a discrete 
system can be bilinearised and its bilinear form arises as a reduction of the Hirota--Miwa equation. However, it was not clarified how the full lattice BSQ 
equation, i.e. the trilinear equation \eqref{dBSQ}, is related to the pentagram map. 

\paragraph{Acknowledgements.}
WF would like to thank Allan Fordy for an introduction to the continuous Sawada--Kotera and Kaup--Kupershmidt equations from the view point of Lax structure, 
and he was supported by a Leeds International Research Scholarship (LIRS) as well as a small grant from the School of Mathematics. 
FWN was partially supported by EPSRC (Ref. EP/I038683/1).
Both authors are grateful to Dajun Zhang for the hospitality during their visit at Shanghai University where the project was initiated. 

\end{document}